\documentclass[11pt]{article}
\usepackage{amsmath, amssymb, amsfonts, times, enumerate}
\usepackage{epsfig, hyperref}
\usepackage{color}
\newcommand{\mb}[1]{{\mathbf{#1}}}
\newcommand{\comment}[1]{}
\newcommand{\Ex}{{\mathbb{E}}}
\DeclareMathOperator\Ind{{{\mathbb I}nd}}
\DeclareMathOperator\supp{supp}
\newcommand{\Prob}{{\mathbb{P}}}
\newcommand{\cl}[1]{\langle #1 \rangle}

\comment{
\renewcommand{\cite}[1]{[#1]}
\def\beginrefs{\begin{list}%
        {[\arabic{equation}]}{\usecounter{equation}
         \setlength{\leftmargin}{2.0truecm}\setlength{\labelsep}{0.4truecm}%
         \setlength{\labelwidth}{1.6truecm}}}
\def\endrefs{\end{list}}

}

\newtheorem{theorem}{Theorem}

\newtheorem{claim}[theorem]{Claim}
\newtheorem{corollary}[theorem]{Corollary}
\newtheorem{definition}[theorem]{Definition}
\newenvironment{proof}{{\bf Proof:}}{\hfill\rule{2mm}{2mm}}

\title{A computational method for bounding the probability of reconstruction on trees}
\author{Nayantara Bhatnagar\thanks{Department of Statistics,
    University of California at Berkeley. ({\tt
      nayan@stat.berkeley.edu}) Supported by DOD ONR
grant N0014-07-1-05-06 and DMS 0528488.} \and 
Elitza Maneva\thanks{Universitat Polit\`ecnica de Catalunya
({\tt maneva@lsi.upc.edu}); 
work done in part while this author was a postdoctoral researcher at IBM Almaden. }}

\begin{document}

\maketitle

\begin{abstract}
For a tree Markov random field non-reconstruction is said to hold if
as the depth of the tree goes to infinity the information that a
typical configuration at the leaves gives about the value at the root
goes to zero. The distribution of the measure at the root
conditioned on a typical boundary can be computed using a
distributional recurrence.
However the exact computation is not feasible because the support of
the distribution grows exponentially with the depth. 

In this work, we introduce a notion of a {\em survey} of a
distribution over probability vectors 
which is a succinct representation of the true
distribution. We show that a survey of the distribution of the measure
at the root can be
constructed by an efficient recursive algorithm. The key properties of
surveys are that the size does not grow with the depth, they
can be constructed recursively, and they still provide a good bound
for the distance between the true conditional distribution and the
unconditional distribution at the root. This approach applies to a
large class of Markov random field models including randomly generated
ones.  
As an application we show bounds on the
reconstruction threshold for the Potts model on small-degree trees.

\end{abstract}


\newpage

\section{Introduction}

Correlations between distant elements or sets of elements in randomly
generated Markov random fields (MRFs) are a main consideration in the
analysis of random constraint satisfaction problems, and in the design
and analysis of message-passing, local search, and other iterative
algorithms for these problems. Here we study one such concept of
correlation on tree MRFs. The presence of this correlation is known as
the property of {\it reconstruction}, and its absence as {\it
  non-reconstruction}. Non-reconstruction is equivalent to the
free-boundary Gibbs measure of the tree being {\it extremal}
\cite{Geobook}. From the 
point of view of statistical physics, reconstruction is equivalent to
{\it replica-symmetry breaking} \cite{MM06}. Some important results on
reconstruction are \cite{KS66, EKPS00, Mos01, MP03, BCMR06, BVVW08, Sly08,
Sly09}. For the connection between this concept and the design of algorithms
for constraint satisfaction and optimization problems we refer the reader to
\cite{MPZ02, ZK07, MSW07} and the references therein. 

\comment{
Our contribution is a new general computational method for
proving that a tree Markov random field has the non-reconstruction
property. We illustrate the method with an application to the Potts
model.   
}
Our contribution is the first general efficient computational method to
obtain non-trivial bounds for 
the probability of reconstruction for a tree Markov random field. 
We illustrate the method with an application to the Potts model.   

Consider a tree MRF.      
The distribution of interest is
the conditional distribution at the root, given a boundary that is
generated randomly according to the MRF. It is said that the root {\it
  cannot be reconstructed}, or the non-reconstruction property holds,
if and only if with high probability this distribution converges to
the unconditional marginal distribution at the root as the depth of
the tree goes to infinity. The distribution of the conditional
distribution at the root (conditioned on a random boundary) can be
expressed recursively\footnote{Notice that the ``distribution of
  ... distribution'' is not a mistake - the object we are interested
  in is indeed the distribution of the randomness left in the root
  after looking at the boundary.} - there is a simple analytic
expression for the distribution for a tree of depth $n+1$ in terms of
the distribution for the tree of depth $n$. However as the depth of
the tree increases the support of the distribution grows
exponentially, which makes numerical estimates difficult to
obtain. For analytical analysis one essentially  needs to find a
special parameter of the distribution which can be bounded
recursively, and this depends on the particular 
MRF.
Often the
analysis has two steps - first showing that the expected distance from
the unconditional distribution is below some small constant, and then
showing that arriving below this small constant implies that this
distance (or some parameter related to it) decreases
geometrically. Particular examples of this approach are the recent 
results on reconstruction for colorings of 
Bhatnagar, Vera, Vigoda, and Weitz \cite{BVVW08} and Sly \cite{Sly08}, 
and for the q-state Potts model of Sly \cite{Sly09}.
The moral from their analysis is that the first step is
the more difficult to achieve. In particular, it is the step that makes 
the analysis possible 
only for the case of large-degree trees. Here we propose a very general method
for making this first step, independent of the particular 
method for generating the MRF, and which is practical for the case of small-degree trees. 

We introduce the notion of a {\it survey} of a distribution over
probability vectors. The survey can be thought of as a "projection"
onto a small "basis" set of probability vectors that carries enough
information about the true distribution. In particular, for the
reconstruction problem, when the goal is to bound from above the
probability of reconstruction, we show that it suffices to keep a
survey of the distribution at each iteration. Applying the recursion
to the surveys for enough iterations allows us to obtain very good
(possibly arbitrarily good) bounds on the probability of
reconstruction.  

We apply our method to the symmetric 3-state Potts model with various parameters in
order to compare with previous results, although it will be clear that
the method does not depend on the symmetry of the model. For the
Potts model the second step is also easy to achieve using the ideas
from recent work of Sly \cite{Sly09}. Since the complexity of the
method is exponential in the degree of the tree, and in the alphabet
of the 
MRF, we do not apply it to very large parameters. However, we
are able to demonstrate new bounds for 3-spins on the 2-ary and 3-ary tree, 
and on a random tree in which every internal vertex has either 2 or 3 
children with equal probability, improving the
bounds of Mossel and Peres \cite{MP03}. We should point out that very recently Formentin and K\"ulske \cite{FK} demonstrated
even better results for this model. Thus, rather than the numerical
results, the importance of our contribution is in the generality of
our method. 

The algorithm we present can be viewed as a rigorous alternative to
the population dynamics algorithm used in physics \cite{MM06, ZK07} 
to determine the spin-glass transition. In population
dynamics the distributional recursion is approximated by keeping a
sample of the distribution at each step. What is not known rigorously
is whether  applying the recursion to a sample of the distribution for
the tree of depth $n$ really results in some sense in a ``faithful''
sample of the distribution for the tree of depth $n+1$. In contrast,
the main technical lemma of our work is precisely the statement that
applying the recursion to the survey of the distribution for the tree
of depth $n$ results in a survey of the distribution for the tree of
depth $n+1$.  

Another related algorithm is the density evolution algorithm for
analysis of the probability of bit-error of LDPC codes
\cite{RU01, RUbook}. The density evolution algorithm, as its name indicates,
is also a recursion on distributions, and in practice it is carried
out by heuristically quantizing the distribution at every step by
rounding. Unfortunately, unlike the reconstruction recursion, the
density evolution recursion does not commute with taking surveys of
the distributions therefore our method cannot be applied, or at least
not in the obvious way. 

Quantization of distributions is also an important design step in the
Survey Propagation algorithm of \cite{MPZ02}. In its analytical form
Survey Propagation uses messages that are distributions with growing
support, whereas in its practical form the messages are distributions
with support of size 3. However we do not know if there is a precise
mathematical connection between the two kinds of quantization. 

This article is organized as follows. In the next section we give the
technical definitions related to reconstruction, and the recursion
relation that we analyze. Section \ref{sec:surveys} is dedicated to
the new definition of a survey of a distributions, its key properties,
and how it can be applied to the reconstruction problem. In the last section we discuss
an application to the Potts model, and compare the results obtained by
our method to previous bounds. 



\section{Reconstruction on trees}

The reconstruction problem in its simplest form can be stated in terms
of a broadcast problem.
Consider a process in which information is broadcast from the root of
an infinite rooted tree $T$ to other vertices as follows. Each edge $e = (u,v)$
acts as a channel $M$ with a finite alphabet $D=\{1, 2, \dots, q\}$. The
channel $M$ is a Markov chain where $(M)_{i,j} = \Pr(v=j|u=i)$.

The letter at the root $\rho$, denoted by $\sigma_\rho$,  is chosen
according to some initial   
distribution. This value is then propagated in the tree as 
follows. For vertex $v$ with parent $u$, let $\sigma_v =
M(\sigma_u)$ for each edge independently. 

For distributions $\mu$ and $\nu$ on the same space $\Omega$, the {\em
  total variation distance} between $\mu$ and $\nu$ is
\begin{eqnarray*}
d_{TV}(\mu,\nu)  = \frac{1}{2} \displaystyle\sum_{\sigma \in \Omega}
|\mu(\sigma) - \nu(\sigma)|
\end{eqnarray*}

Let $L(n)$ denote the configuration at level $n$.
\begin{definition}The tree $T$ and channel $M$ have the {\em
    non-reconstruction property} if for every $a,b \in D$,
\begin{eqnarray*}
\lim_{n \rightarrow \infty} d_{TV} [L^a(n),L^b(n)] = 0
\end{eqnarray*}
where $L^c(n)$ denotes the conditional distribution of $L(n)$ given
that $\sigma_\rho = c$.
\end{definition}

In 2006 M\'ezard and Montanari \cite{MM06} showed that the same problem
had also been studied in physics for even more general models, and the
reconstruction transition is the equivalent to the {\em replica symmetry
breaking spin-glass transition}. In the physics formalism rather than
using channels, the system is described as a Markov random field,
which in many application is also randomly generated, and the question
again is whether the root is reconstructible from typical
assignments at the leafs.  

Here, for greatest generality we will define the reconstruction
problem using a Markov random field (MRF). This way we capture models
such as 3-SAT and other constraint satisfaction problems, which are
not as easy to think about in terms of channels.   

Consider a tree $T=(V, E)$ with root vertex $\rho$. For every $v \in
V$ there is a domain of values $D_v$ that can be assigned to this
vertex, and for every edge $e=(u,v)$ a non-negative function we call
potential $\Psi_e:D_u \times D_v \rightarrow R^+$. 
We will define a distribution over assignments to the vertices
$\times_{v\in V} D_v=D$. For every configuration $\overline{\sigma} \in D$,
let $\sigma_v$ denote the component corresponding to vertex  
$v$, and $\overline{\sigma}(n)$ the components corresponding to level $n$
of the tree. 
\begin{equation}
\label{eq:MRF}
\Prob_{T, \Psi}(\overline{\sigma}):= \frac{1}{Z_{T, \Psi}}
\prod_{e=(u,v)\in E} \Psi_e(\sigma_u, \sigma_v), 
\end{equation}
where $Z$ is the normalizing constant $Z_{T, \Psi}=\sum_{\overline{\sigma}
  \in D} \prod_{e=(u,v) \in E} \Psi_e(\sigma_u, \sigma_v)$.  

We will allow random MRFs $(T, \Psi)$ generated in the following
way. To each level there corresponds a degree distribution and a
distribution over potential functions. A tree of depth $1$ is just a
single vertex (no potential functions). An MRF of depth $n+1$ is
generated by choosing a degree for the root from the degree
distribution of level $n+1$, next choosing potential functions for all the
edges adjacent to the root independently from the distribution of
potential functions corresponding to level $n+1$, and finally
attaching randomly generated MRFs of depth $n$ to the other ends of
the edges. 

We denote by $L(n)$ the configuration at the vertices at level $n$
(assuming $T$ is of depth at least $n$), and by $\rho$ the root. Let
also $L^a(n)$ denote the configuration at level $n$ conditional on the
root having value $a\in D_\rho$. 

\begin{definition} We say that a random MRF 
{\em has the non-reconstruction property} if for every  $a, b \in D_\rho$  
$$ \lim_{n\rightarrow \infty} \Ex_{T, \Psi}[d_{TV}(L^a(n), L^b(n))]=0.$$
\end{definition}

For the rest of this section we will consider a fixed MRF, so we can
omit $\Psi$ from the subscript as the functions on the edges will not
be changing, and we will also omit $T$ from the subscript whenever it
is understood. We will also use $a$ and $b$ for the events that the
root takes the value $a$ or $b$ respectively.  

We will use an alternative expression for the total variation distance,
which follows from Bayes' rule. For this we denote $\pi(a) :=
\Prob_T(a)$ and $\pi(b):=\Prob_T(b)$.  
\begin{eqnarray}
d_{TV}(L^a(n), L^b(n))&=&  \sum_L | \Prob(L(n)=L | a) - \Prob(L(n)=L | b) |\\
&=& \sum_L \left| \frac{\Prob(a |L(n)=L) \Prob(L(n)=L)}{\pi(a)} -
  \frac{\Prob(b | L(n)=L)\Prob(L(n)=L)}{\pi(b)}\right|\\ 
&=&\Ex_{L(n)}\left[\left|\frac{\Prob(a|L(n))}{\pi(a)} -
    \frac{\Prob(b|L(n))}{\pi(b)}\right|\right]\\ 
&\le&\frac{1}{\pi(a)}\Ex_{L(n)}\left[\left| \Prob(a|L(n))-\pi(a)\right|\right] +
\frac{1}{\pi(b)}\Ex_{L(n)}\left[\left| \Prob(b|L(n))-\pi(b)\right|\right].
\label{eq:ineq}
\end{eqnarray}

An immediate observation is that this quantity is monotonically decreasing with the depth of the tree $n$.
\begin{eqnarray*}
&&\Ex_{L(n+1)}\left[\left|\frac{\Prob(a|L(n+1))}{\pi(a)} -
    \frac{\Prob(b|L(n+1))}{\pi(b)}\right|\right] 
\\
&&=\Ex_{L(n+1)}\left[\left|\sum_{L} 
\frac{\Prob(a|L(n)=L) ~\Prob(L(n)=L|L(n+1))}{\pi(a)} - 
\frac{\Prob(b|L(n)=L) ~\Prob(L(n)=L|L(n+1))}{\pi(b)}
\right|\right]\\
&&\le
\Ex_{L(n+1)}\left[\sum_{L} \Prob(L(n)=L | L(n+1)) \times 
\left|
\frac{\Prob(a|L(n)=L)}{\pi(a)} - \frac{\Prob(b |L(n)=L)}{\pi(b)}
\right|\right]\\
&&=
\sum_{L'} \Prob(L(n+1)=L') \sum_{L} \Prob(L(n)=L | L(n+1)=L') \times 
\left|
\frac{\Prob(a|L(n)=L)}{\pi(a)}- \frac{\Prob(b|L(n)=L)}{\pi(b)}\right|\\
&&=
\sum_L \Prob(L(n)=L)
\left|
\frac{\Prob(a|L(n)=L)}{\pi(a)}- \frac{\Prob(b |L(n)=L)}{\pi(b)}\right|\\
&&= \Ex_{L(n)}\left[\left|\frac{\Prob(a|L(n))}{\pi(a)} - \frac{\Prob(b|L(n))}{\pi(b)}\right|\right].
\end{eqnarray*}

This expression for the variation distance is a function of a
distribution that we will refer to many times in the rest of the
paper, so for convenience we define the following terminology for it:  

\begin{definition} For a tree Markov random model as in equation
  \eqref{eq:MRF} the {\em residual distribution at the root of T} is
  the distribution of the marginal distribution at the root
  conditional on a random boundary, which is chosen according to the
  distribution $\Prob_{T,\Psi}$. In other words it is the distribution
  of the probability vector $\eta=(\eta^a:=\Prob_T(a | L(n)), a \in
  D_\rho)$, where $L(n)$ is a configuration for the leaves chosen
  according to $\Prob_{T, \Psi}$. 
\end{definition}

Intuitively, a random MRF has the non-reconstruction property if the residual distribution at the root is with high probability concentrated on probability vectors arbitrarily close to the stationary distribution $\pi$. In particular, we will aim to show that 
$\Ex[| \Prob(a | L(n)) -\pi(a) |]$ goes to 0 for every $a\in D_\rho$, which implies non-reconstruction by the inequality \eqref{eq:ineq}.  

Next we derive the recursive equation for the residual distribution at the root of a tree.
For a tree $T$ of depth $n$ with edges $E_T$ let 
$$Z_T(\overline{\sigma}):= \prod_{e=(u,v) \in E_T} \Psi_e(\sigma_u, \sigma_v).$$ 
Thus $Z_T:=\sum_{\overline{\sigma}}Z_T(\overline{\sigma})$. For a boundary configuration $L$, let 
$$Z_T(L):=\sum_{\overline{\sigma}: \overline{\sigma}(n)=L} Z_T(\overline{\sigma}).$$ 
For a boundary $L$ let $\eta_T(L)=(\eta^a_T(L): a\in D_\rho)$ denote the probability vector for the distribution of the assignment at the root conditional on the boundary being $L$, i.e. $\eta^a_T(L):=\Prob_T(a|L(n)=L)$. For any probability vector $\eta$ on $D_\rho$ let 
$$Z_T(\eta):= \sum_{L:\eta_T(L)=\eta} Z_T(L).$$ 
Thus the probability that the marginal distribution at the root 
is equal to $\eta$ is $Z_T(\eta)/Z_T$. 

Suppose the children of the root of T are $v_1, \dots, v_r$, the corresponding edges connecting them to the root are $e_1, 
\dots, e_r$ and the subtrees rooted at them are respectively $T_1, \dots, T_r$. Consider a boundary configuration  $L=(L_1, \dots, L_r)$ for the large tree $T$, where $L_i$ is the part of the boundary that belongs to $T_i$. It is straightforward to derive the expression for $\eta_T(L)$ in terms of $\eta_{T_i}(L_i)$:
\begin{eqnarray*}
\eta_T^a(L) &=& 
\frac{  \sum_{\overline{\sigma}: \overline{\sigma}(n)=L, \sigma_\rho=a}  Z(\overline{\sigma}) }
{Z_T(L)}\\
&=& 
\frac{\sum_{\sigma_{v_1}, \dots, \sigma_{v_r}} 
\prod_{i=1}^r  \Psi_{e_i}(a, \sigma_{v_i}) ~ Z_{T_i}(L_i) \eta^{\sigma_{v_i}}_{T_i}(L_i) }
{Z_T(L)} \\
&=& 
\frac{\prod_{i=1}^r Z_{T_i}(L_i)}{Z_T(L)} 
\prod_{i=1}^r ~\sum_{\sigma_{v_i \in D_{v_i}}} 
\Psi_{e_i}(a, \sigma_{v_i}) ~\eta^{\sigma_{v_i}}_{T_i}(L_i) 
\end{eqnarray*}

It is convenient to define the corresponding function (which is actually the update function of the belief propagation algorithm for computing the marginal distribution at the root of an acyclic MRF). Let ${\cal P}_v$ be the space of probabilities over the domain $D_{v}$ for any $v \in V$. We define $f: {\cal P}_{v_1}, \dots, {\cal P}_{v_r} \rightarrow {\cal P}_\rho$ in the following way:
$$ f^a (\eta_1, \eta_2, \dots, \eta_r):= 
\prod_{i=1}^r ~~\sum_{\sigma_{v_i} \in D_{v_i}}
 \eta_i^{\sigma_{v_i}}  \Psi_{e_i}( a, \sigma_{v_i}).$$
We will also need to use the norm $\| f(\eta_1, \dots, \eta_r)\| = \sum_{a \in D_\rho} f^a(\eta_1, \dots, \eta_r).$ With this notation
$$\eta_T^a(L) = \frac{f^a( \eta_{T_1}(L_1), \dots, \eta_{T_r}(L_r))}
{\|f(\eta_{T_1}(L_1), \dots, \eta_{T_r}(L_r))\|}.$$
For two vectors  the symmetric relation $\propto$ will indicate that the vectors are equal up to a multiplicative constant. Thus $\eta_T(L)\propto f( \eta_{T_1}(L_1), \dots, \eta_{T_r}(L_r))$. When used with probabilities it will indicate that the normalization constant has been omitted. We are now ready to derive the recursion for the residual distribution at the root of the tree.

\begin{theorem} 
\label{thm:recon}
Let $\mb{P_i}$ and $\mb{Q}$ be random vectors such that $\mb{P_i}$ is 
distributed according to the residual distribution at the root of $T_i$ for $i=1, \dots, r$ and $\mb{Q}$ is distributed according to the residual distribution at the root of $T$. Then for any $\eta \in {\cal P}_\rho$
$$\Prob(\mb{Q}=\eta) \propto \Ex[\|f(\mb{P_1}, \dots, \mb{P_r})\| \times \Ind[f(\mb{P_1}, \dots, \mb{P_r})\propto \eta]].$$

\end{theorem}
\begin{proof}
First we derive the recursion for the total weight of configurations with a fixed boundary:
\begin{eqnarray*}
Z_T(L)&=& 
\sum_{a \in D_\rho} \sum_{\sigma_{v_1}, \dots, \sigma_{v_r}} 
\prod_{i=1}^r  \Psi_{e_i}(a, \sigma_{v_i}) ~ Z_{T_i}(L_i) \eta^{\sigma_{v_i}}_{T_i}(L_i) \\
&=& 
\left(\prod_{i=1}^r Z_{T_i}(L_i)\right) \times \sum_{a \in D_\rho} ~
\prod_{i=1}^r\sum_{\sigma_{v_i} \in D_{v_i}} \eta^{\sigma_{v_i}}_{T_i}(L_i) ~\Psi_{e_i}(\sigma_\rho, \sigma_{v_i})\\
&=& 
\left(\prod_{i=1}^r Z_{T_i}(L_i)\right) \times \| f(\eta_{T_1}(L_1), \dots, \eta_{T_r}(L_r))\|.
\end{eqnarray*}
Next, using the above, we derive the recursion for the total weight of configurations yielding a given 
marginal distribution at the root:
\begin{eqnarray*}
Z_T(\eta)
&=& 
\sum_{L: \eta_T(L)=\eta} Z_T(L)\\
&=&
\sum_{L=(L_1, \dots, L_r): \atop f(\eta_{T_1}(L_1), \dots, \eta_{T_r}(L_r)) \propto \eta} 
\left(\prod_{i=1}^r Z_{T_i}(L_i)\right) \times \| f(\eta_{T_1}(L_1), \dots, \eta_{T_r}(L_r))\|\\
&=&
\sum_{\eta_1, \dots, \eta_r: \atop  f(\eta_1, \dots, \eta_r)\propto \eta} 
\|f(\eta_1, \dots, \eta_r)\| \sum_{L_1: \eta_{T_1}(L_1)=\eta_1} \dots \sum_{L_r: \eta_{T_r}(L_r)=\eta_r} 
\prod_{i=1}^r   Z_{T_i}(L_i)\\
&=&  
\sum_{\eta_1, \dots, \eta_r:  \atop f(\eta_1, \dots, \eta_r)\propto \eta} 
\|f(\eta_1, \dots, \eta_r)\| \prod_{i=1}^r Z_{T_i}(\eta_i).
\end{eqnarray*}
Finally we can derive the recursion for the residual distribution at the root:
\begin{eqnarray*}
\Prob(\mb{Q}=\eta)&=&\frac{Z_T(\eta)}{Z_T} \\
&=& 
\frac{\prod_{i=1}^r Z_{T_i}}{Z_T}
\sum_{\eta_1, \dots, \eta_r: \atop f(\eta_1, \dots, \eta_r)\propto \eta} 
\|f(\eta_1, \dots, \eta_r)\| \prod_{i=1}^r \frac{Z_{T_i}(\eta_i)}{Z_{T_i}}\\
&=&
\frac{\prod_{i=1}^r Z_{T_i}}{Z_T} ~~
\Ex[ \|f(\mb{P_1}, \dots, \mb{P_r})\| \times \Ind[f(\mb{P_1}, \dots, \mb{P_r})\propto \eta]]. 
\end{eqnarray*}
\end{proof}

This recursion was also derived by M\'ezard and Montanari in \cite{MM06} and, as pointed out by them, its fixed point is known as the ``1-step Replica Symmetric Breaking solution with Parisi parameter $m=1$'' (in the general 1RSB scheme the factor $\|f(\dots)\|$ is raised to a power $m$). 
Almost all other work in reconstruction considers, instead, the recursion for the {\em conditional} residual distribution at the root, conditioned on the boundary being generated from the MRF with a fixed value at the root. The method we present here can be applied only with the unconditional distribution. 

The main contribution of this article is a method for discretizing the recursion of Theorem \ref{thm:recon}. The only property of $f$ that is used is that it is a multi-affine function (i.e. affine in each coordinate). Thus the main technical theorem will not use the definitions related to reconstruction.

\section{Surveys of distributions}
\label{sec:surveys}

Let us first illustrate our approach with the example of the channel
corresponding to random coloring on the  regular $d$-ary tree. Each vertex can take one of $q$
colors. Based on the color of the parent, the colors of the children
are chosen independently and uniformly at random from the set of
colors different from that of the parent. If we start with a uniformly
random color at the root, this process generates a random coloring of
the tree. Showing non-reconstruction for this process is equivalent to
showing that for almost all colorings of the leafs generated by
choosing a random coloring of the tree, there are almost the same
number of colorings consistent with these leafs, in which the root of
the tree has each of the $q$ colors.  More precisely, let $C_{L}$
denote the colorings of a tree of depth $n$ with leafs colored
according to $L$, and $C^i_{L}$ denote the colorings in $C_L$
with color $i$ at the root. We will recursively try to match up the
colorings in the sets $C^1_L, ..., C^q_L$ by splitting them into sets
that are balanced or close to balanced. For example, suppose $q=3$,
the colors are called Red, Blue and Green, and for some $L$ we have
$(C^R_L, C^B_L, C^G_L)=(10,5,7)$.  Then we can split this set of
colorings into 3 sets : one that is perfectly balanced, containing 5
colorings of each type; one that is balanced with respect to Red and
Green, containing two colorings with Red at the root, and two with
Green at the root; and one that contains the remaining 3 colorings
with Red at the root ($(10,5,7)=(5,5,5)+(2,0,2)+(3,0,0)$). We call
such sets \emph{bundles}. A bundle always contains only colorings that
have the same colors at the leafs. Bundles can be of several different
types according to the ratio of number of colorings with red, blue and green at the root,  
and we can choose the types that will be allowed. 
We will only keep track of the number of
bundles of each type.  The goal is to construct the bundles
recursively in such a way that the majority of colorings are eventually  in
balanced bundles.

There is a simple process to construct bundles recursively. Suppose we
have a particular splitting of the colorings of the $d$-ary tree of depth $n$
into bundles. For any $d$-ple of these bundles, consider the set
of colorings of the tree of depth $n+1$ such that the first subtree of
depth $n$ is colored with a coloring from the first bundle, the second
subtree with a coloring from the second bundle, etc. The resulting set
of colorings on the depth $n+1$ tree, has the following properties:
(1) all colorings have the same colors at the leafs, and (2) the
number (or fraction) of colorings in this set with a specific color at the root can
be computed exactly using only the types of the bundles in the
$d$-ple. The resulting set of colorings on the depth $n+1$ tree
can be split again into bundles of the allowed types. The specific
recipe for splitting this set into bundles of course will influence to
what extent the balanced and near-balanced bundles dominate the entire
collection of bundles.

In the special case that the bundles are defined to be the set of colorings  with a given boundary, i.e. all types of bundles are allowed and no splitting of bundles occurs, this gives exact recursive computation of the distribution of the marginal distribution at the root.

In this next section we formalize and generalize the above construction.

\subsection{Definitions}


Let ${\cal V}$ denote a real vector space of finite dimension.
Let ${\cal S} = (S_1,\dots, S_n)$ be a finite sequence with $S_i \in {\cal V}$. 
We denote the convex hull of ${\cal S}$ in ${\cal V}$ by
$\cl{\cal S}$. Let $\alpha_1, \dots, \alpha_n$ be functions from the
convex hull $\cl{\cal S}$ to $[0,1]$ such that the following
properties hold for every $\eta \in \cl{\cal S}$:
\begin{enumerate}
\item $\sum_{i=1}^n \alpha_i(\eta)=1$, 
\item $\eta=\sum_{i=1}^n \alpha_i(\eta)~S_i$.
\end{enumerate}
Thus for every $\eta\in \cl{\cal S}$ these functions define a convex
decomposition of $\eta$. We will
call the tuple $({\cal S}, \alpha_1, \dots, \alpha_n)$ a
\emph{skeleton} in ${\cal V}$, and ${\cal S}$ the \emph{base set} of
the skeleton.

By $\mb{P}$ we denote a distribution on ${\cal V}$ with finite support
as well as a random vector chosen from this distribution. 
Let the skeleton $({\cal S}, \alpha_1, \dots,
\alpha_n)$ be such that the support of the distribution of $\mb{P}$ lies inside $\cl{\cal
S}$. Let $\mb{C}$ be a random element of ${\cal S}$ with the
following distribution: $\Prob(\mb{C}=S_i)=\Ex[\alpha_i(\mb{P})]$. This is well defined by the first condition on the functions $\alpha_1, \dots,
\alpha_n$. We will call $\mb{C}$ a \emph{survey} of $\mb{P}$ on the skeleton $({\cal S},
\alpha_1, \dots, \alpha_n)$.  

We say that $\mb{A}$ is a survey of $\mb{B}$ without specifying the
skeleton, whenever there exists a skeleton with respect to which
$\mb{A}$ is the survey of $\mb{B}$.


\subsection{Properties of surveys}

In this section we show several useful properties of surveys.

\begin{theorem}
\label{thm:transitivity}
If $\mb{P}$ is a distribution on ${\cal V}$ with finite support, 
$\mb{C}$ is a survey of $\mb{P}$, and $\mb{D}$ is a
survey of $\mb{C}$, then $\mb{D}$ is a survey of $\mb{P}$.
\end{theorem}

\begin{proof}
Suppose 
$\mb{C}$ is a survey of $\mb{P}$ on a skeleton $({\cal S}, \alpha_1, \dots, \alpha_n)$, and 
$\mb{D}$ is a survey of $\mb{C}$ on a skeleton $({\cal T}, \beta_1, \dots, \beta_m)$. 
Let $\gamma_i(\eta):=\sum_{j=1}^n \alpha_j(\eta) \beta_i(S_j)$ for $i=1, \dots, n$. 
Then $({\cal T}, \gamma_1, \dots, \gamma_m)$ is a valid skeleton, because for every 
$\eta\in {\cal V}$ 
$$\sum_{i=1}^m \gamma_i(\eta) =\sum_{i=1}^m\sum_{j=1}^n \alpha_j(\eta) \beta_i(S_j)=
 \sum_{j=1}^n \alpha_j(\eta)  \sum_{i=1}^m \beta_i(S_j) = 
 \sum_{j=1}^n \alpha_j(\eta)  = 1,$$
$$\sum_{i=1}^m \gamma_i(\eta) T_i =
\sum_{i=1}^m \sum_{j=1}^n \alpha_j(\eta) \beta_i(S_j) T_i = 
\sum_{j=1}^n \alpha_j(\eta) ~ \sum_{i=1}^n \beta_i(S_j) T_i = 
\sum_{j=1}^n \alpha_j(\eta) S_j = \eta.$$
Finally, we can verify that $\mb{D}$ is a survey of $\mb{P}$ on the skeleton $({\cal T}, \gamma_1, \dots, \gamma_m)$.
\begin{eqnarray*}
\Prob(\mb{D}=T_i) &=& \Ex[\beta_i(\mb{C})] = 
\sum_{j=1}^n \Prob(\mb{C}=S_j) ~ \beta_i(S_j) 
= \sum_{j=1}^n \Ex[\alpha_j(\mb{P})] ~ \beta_i(S_j) \\
&=& \Ex\left[ \sum_{j=1}^n \alpha_j(\mb{P}) \beta_i(S_j)\right] 
= \Ex[\gamma_i(\mb{P})]
\end{eqnarray*}
\end{proof}

\begin{theorem}
\label{thm:mix}
Let $\mb{P_1}$ and $\mb{P_2}$ be independent distributions on ${\cal V}$ with finite support, 
and let $\mb{C_1}$ and
$\mb{C_2}$ be their surveys. Suppose a distribution $\mb{P}$ on ${\cal V}$ is defined to be equal to $\mb{P_1}$ with probability $p\ge 0$ and $\mb{P_2}$ with
probability $1-p$, and similarly $\mb{C}$ is defined to be $\mb{C_1}$
with probability $p$ and $\mb{C_2}$ with probability $1-p$. Then
$\mb{C}$ is a survey of $\mb{P}$.
\end{theorem}

\begin{proof}
Suppose $\mb{C_1}$ and $\mb{C_2}$ are surveys respectively of $\mb{P_1}$ and $\mb{P_2}$ 
on skeletons $({\cal S}, \alpha_1, \dots, \alpha_n)$, and $({\cal T}, \beta_1, \dots, \beta_m)$.  
Let ${\cal R}={\cal S} \cap {\cal  T}$ and $|{\cal R}|=k$. Without loss of generality, let's suppose that 
${\cal R}= \{S_1, \dots, S_k \}$ and $S_i=T_i$ for $i=1, \dots, r$.  Let's denote the union of the two basis sets as ${\cal U} = (U_1, \dots, U_{m+n-k}) = (S_1, \dots, S_n, T_{r+1}, \dots, T_m )$. This set will be the basis of the skeleton. Next we define the functions:
\begin{equation*}
\gamma_j(\eta) :=\frac{1}{\Prob(\mb{P}=\eta)} \times 
\begin{cases} 
p ~ \alpha_j(\eta)~ \Prob(\mb{P_1}=\eta) + (1-p) ~\beta_j(\eta) ~\Prob(\mb{P_2}=\eta)& j = 1, \dots,  r\\
p~\alpha_j(\eta) ~ \Prob(\mb{P_1}=\eta) & j = r+1, \dots, n\\
(1-p)~\beta_{j-(n-r)}(\eta) ~\Prob(\mb{P_2}=\eta) & j = n+1, \dots, m+n-k\\
\end{cases}
\end{equation*}
It is  straightforward to check that $({\cal U}, \gamma_1, \dots, \gamma_{m+n-k} )$ is a valid skeleton. To check that $\mb{C}$ is a survey of $\mb{P}$ on this skeleton we just need $\Prob(\mb{C}=U_j) = \Ex[\gamma_j(\mb{P})]. $\ For $j=1, \dots, r$
\begin{eqnarray*}
\Prob(\mb{C}= U_j) &=& p ~\Prob(\mb{C_1}=U_j) + (1-p) ~\Prob(\mb{C_2}=U_j) \\
&=&  p ~\Ex[\alpha_j(\mb{P_1})] + (1-p)~\Ex[\beta_j(\mb{P_2})] \\
&=& \sum_{\eta \in {\cal P}}  \left(p~ \Prob(\mb{P_1}=\eta) ~\alpha_j(\eta) +
(1-p) ~\Prob(\mb{P_2}=\eta)~\beta_j(\eta) \right) \\ 
&=& \sum_{\eta \in {\cal P}} \Prob(\mb{P}=\eta) ~ \gamma_j(\eta) \\
&=&\Ex[\gamma_j(\mb{P})].
\end{eqnarray*}
Similarly, for $j=r+1, \dots, n$, and for $j=n+1, \dots, n+m-k$ we have respectively
\begin{eqnarray*}
\Prob(\mb{C}= U_j) =& 
p ~\Prob( \mb{C_1}=U_j)~~~~~~~~ = 
~~~~~~~~p ~ \Ex[\alpha_j(\mb{P_1})] &= 
 \Ex[\gamma_j(\mb{P})], \\
 \Prob(\mb{C}= U_j) =& 
 (1-p) ~\Prob( \mb{C_2}=U_{j-(n-r)}) = 
 (1-p) ~\Ex[\beta_{j-(n-r)}(\mb{P_2})] &= 
 \Ex[\gamma_j(\mb{P})],
 \end{eqnarray*} 
\end{proof}

The next theorem is the one that allows us to use surveys in the context of the reconstruction recursion of Theorem \ref{thm:recon}. Let ${\cal V}_1, \dots, {\cal V}_r$ denote real vector spaces of finite dimensions. We say that a function 
$f: {\cal V}_1 \times \dots \times {\cal V}_r  \rightarrow {\cal V}$ is {\em multi-affine} if for every 
$\eta_1 \in {\cal V}_1, \eta_2 \in {\cal V}_2, \dots, \eta_r \in {\cal V}_r$, $a\ge 0$ and $b \ge 0$, such that $a+b=1$, $i \in \{1, \dots, r\}$, and  $\eta_i' \in {\cal V}_i$,
 it holds that 
 $$f(\eta_1, \dots, a\eta_i+b\eta_i', \dots, \eta_r)=
a f(\eta_1, \dots, \eta_i, \dots, \eta_r)+ b f(\eta_1, \dots,
\eta_i', \dots, \eta_r).$$ Recall also that for a vector $\eta \in {\cal V}$ 
we denote by $\|\eta\|$ the sum of the coordinates of $\eta$. 

\begin{theorem}
\label{thm:main}
Let $f:  {\cal V}_1 \times \dots \times {\cal V}_r  \rightarrow {\cal V}$ 
be a multi-affine function. 
Let  $\mb{P_1}, \dots, \mb{P_r}$ be independent distributions on ${\cal V}_1, \dots, {\cal V}_r$ with finite support,
and $\mb{C_1}, \dots, \mb{C_r}$ be their respective surveys. 
If $\mb{Q}$ and $\mb{D}$ are random elements of ${\cal V}$ defined in the following way:
$$\Prob(\mb{Q}= \eta) \propto \Ex\left[\|f(\mb{P_1}, \dots, \mb{P_r})\| 
\times \Ind\left[f(\mb{P_1}, \dots, \mb{P_r}) \propto \eta \right]\right],$$
$$\Prob(\mb{D}= \eta) \propto \Ex\left[\|f(\mb{C_1}, \dots, \mb{C_r})\| 
\times \Ind\left[f(\mb{C_1}, \dots, \mb{C_r}) \propto \eta \right]\right],$$
then $\mb{D}$ is a survey of $\mb{Q}$.
\end{theorem}

\begin{proof}
We first demonstrate the proof for the case $r=1$. We denote $\mb{P_1}$ by $\mb{P}$ 
and $\mb{C_1}$ by $\mb{C}$.

Suppose $\mb{C}$ is a survey of $\mb{P}$ on the skeleton $({\cal S}, \alpha_1, \dots, \alpha_n)$. 
Consider the set $\{f(S_i)/\|f(S_i)\|: S_i \in {\cal S}\}$. The size of
this set can be less than $n$ if for two different indexes $1\le
i<i'\le n$ it happens that
$f(S_i)/\|f(S_i)\|=f(S_{i'})/\|f(S_{i'})\|$. Let's denote by ${\cal
T}=(T_1, \dots, T_m)$ the ordered list of distinct elements
corresponding to the above set. Let
$I_j = \{i: f(S_i)/\|f(S_i)\|=T_j \}$.

Since the support of $\mb{C}$ is
${\cal S}$ it follows that the support of $\mb{D}$ is contained in
${\cal T}$.  It suffices to find $\beta_1,
\dots, \beta_m$ non-negative functions on ${\cal V}$ such that
$\Prob(\mb{D}=T_j) = \Ex[\beta_j(\mb{Q})]$. We begin with the left-hand side:
\begin{eqnarray*}
\Prob(\mb{D}=T_j) &\propto& \sum_{i=1}^n \Prob(\mb{C}=S_i) 
\times \|f(S_i)\| \times \Ind[f(S_i)/\|f(S_i)\|=T_j] \\
&=& \sum_{i\in I_j} \Prob(\mb{C}=S_i) \times \|f(S_i)\| \\
&=& \sum_{i\in I_j} \Ex[\alpha_i(\mb{P})] \times \|f(S_i)\| \\
&=& \Ex\left[ \sum_{i\in I_j} \alpha_i(\mb{P})~ \|f(S_i)\|\right] 
\end{eqnarray*}
Thus the normalization constant for the above probability is:
\begin{eqnarray*}
\sum_{j=1}^m \Ex\left[\sum_{i\in I_j} \alpha_i(\mb{P}) ~ \|f(S_i)\|\right] 
=
\Ex\left[\sum_{i=1}^n \alpha_i(\mb{P}) ~ \|f(S_i)\|\right] 
= \Ex\left[\| f \left( \sum_{i=1}^n \alpha_i(\mb{P})S_i\right)\|  \right]
= \Ex[\|f(\mb{P})\|],
\end{eqnarray*}
where the second equality follows from the fact that $f$ is multi-affine.

Next, we look at the right hand-side of the desired equality.  For
every $\eta \in {\cal V}$ we define 
\begin{eqnarray*}
W(\eta) &=& \Ex[\|f(\mb{P})\| \times \Ind[f(\mb{P}) \propto \eta]],\\
W &=& \Ex[\|f(\mb{P})\|].\\
\end{eqnarray*}
Then by the definition of $\mb{Q}$,
$\Prob(\mb{Q}=\eta) = W(\eta)/W$. For every $j\in \{1, \dots, m\}$,
let's define 
$$\beta_j (\eta) = \frac{1}{W(\eta)}~
\Ex\left[\sum_{i \in I_j}\alpha_i(\mb{P}) \times \|f(S_i)\|
\times \Ind[ f(\mb{P}) \propto \eta]  \right].$$
It is easy to check that $\sum_{j=1}^m \beta_j(\eta)=1$, and
$\sum_{j=1}^m \beta_j(\eta)~T_j=\eta$ for every $\eta \in {\cal V}$. Finally,
\begin{eqnarray*}
\Ex[\beta_j(\mb{Q})] &=& 
\sum_{\eta\in \supp(\mb{Q})} \Prob(\mb{Q}=\eta) \times
\frac{1}{W(\eta)}~
\Ex \left[\sum_{i \in I_j} \alpha_i(\mb{P}) \times \|f(S_i)\|
\times \Ind[ f(\mb{P}) \propto \eta] \right] \\
&=&
\sum_{\eta\in \supp(\mb{Q})} \frac{W(\eta)}{W} \times \frac{1}{W(\eta)} ~
\Ex \left[\sum_{i \in I_j} \alpha_i(\mb{P}) \times \|f(S_i)\|
\times \Ind[ f(\mb{P}) \propto \eta] \right]  \\
&=& \frac{1}{W}~ \Ex\left[\sum_{i \in I_j} \alpha_i(\mb{P})~ \|f(S_i)\|\right] 
= \Prob(\mb{D}=T_j)
\end{eqnarray*}

Next we generalize the proof to $r>1$. 
Suppose $\mb{C_1}, \dots, \mb{C_r}$ are surveys on skeletons 
$({\cal S}^1, \alpha^1_1, \dots, \alpha^1_{n_1})$, \dots, 
$({\cal S}^r, \alpha^r_1, \dots, \alpha^r_{n_r})$. 
The same proof applies by defining the base set for the
skeleton to be $$\{ f(S^1_{i_1}, \dots, S^r_{i_r})/\|f(S^1_{i_1}, \dots,
S^r_{i_r})\| : S^1_{i_1} \in {\cal S}^1, \dots, S^r_{i_r}\in {\cal S}^r\}$$ 
and changing the notation to
$I_j = \{(i_1, \dots, i_r): f(S^1_{i_1}, \dots,
S^r_{i_r})/\|f(S^1_{i_1}, \dots, S^r_{i_r})\| = T_j\}$, and
$\mb{P}=(\mb{P_1}, \dots, \mb{P_r})$. We have that

\begin{eqnarray*}
\Prob[\mb{D}=T_j] 
&=& \frac{\Ex\left[ \sum_{(i_1, \dots, i_r)\in I_j} 
\left(\prod_{k=1}^r \alpha^k_{i_k}(\mb{P_k})\right)~ 
\|f(S^1_{i_1}, \dots, S^r_{i_r})\|\right] }
{\Ex[\|f(\mb{P})\|]} = \Ex[\beta_j(\mb{Q})],
\end{eqnarray*}
where
$$\beta_j (\eta) = \frac{1}{W(\eta)}~
\Ex\left[\sum_{(i_1, \dots, i_r) \in I_j} 
\left(\prod_{k=1}^r \alpha^k_{i_k}(\mb{P_k})\right) \times \|f(S^1_{i_1}, \dots, S^r_{i_r})\|
\times \Ind[ f(\mb{P}) \propto \eta]  \right].$$
\end{proof}

\begin{theorem}
\label{thm:convex}
Let $g$ be a convex function on ${\cal V}$.  
For every distribution $\mb{P}$  on  ${\cal V}$ of finite support, 
and $\mb{C}$ a survey of $\mb{P}$, $\Ex[g(\mb{P})] \le \Ex[g(\mb{C})]$ and
the equality holds if $g$ is an affine function.
\end{theorem}
\begin{proof}
\begin{eqnarray*}
\Ex[g(\mb{C})] &=& \sum_{i=1}^n \Prob[\mb{C}=S_i]~g(S_i) 
= \sum_{i=1}^n \Ex[\alpha_i(\mb{P})] ~g(S_i) \\
&=& \Ex \left[\sum_{i=1}^n \alpha_i(\mb{P}) ~g(S_i) \right]
\ge \Ex \left[g\left(\sum_{i=1}^n \alpha_i(\mb{P})~ S_i\right) \right]
= \Ex[g(\mb{P})]
\end{eqnarray*}
\end{proof}

\subsection {Using surveys to bound the probability of reconstruction}

Using the theorems from the previous section now we can show how to calculate recursively a survey of the residual distribution at the root of a random MRF on a skeleton of bounded size. Assume first that the degree distributions and potential function distributions for all levels have small support. Suppose we have a survey $\mb{C}$ for the random MRF of depth $n$. For level $n+1$ first we calculate, for each possible instantiation of the degree $r$ and potential functions $\Psi_1, \dots,\Psi_r$ at the root, a survey of the residual distribution at the root using $r$ copies of $\mb{C}$ and the recursion of Theorem $\ref{thm:recon}$. The result is a survey of the true residual distribution for this instantiation by Theorem \ref{thm:main}. Suppose the probability of this instantiation is $p(r, \Psi_1 , \dots, \Psi_r)$. Next, we combine these distributions by defining a distribution equal to the surveys of  each of the instantiations with probability $p(r, \Psi_1 , \dots, \Psi_r)$. This is a survey of the true residual distribution for the random MRF of depth $n+1$ by 
Theorem \ref{thm:mix}. Finally, if the support of the resulting survey is bigger than the required bound, we can choose a smaller skeleton and compute a survey of the survey, which by Theorem \ref{thm:transitivity} is also a survey of the true residual distribution.

Finally, since distance from a fixed distribution is a convex function, by Theorem 
\ref{thm:convex} the expected distance of the survey of the residual distribution from $
\pi$ is an upper bound on the distance of the true residual distribution from $\pi$. The quality of this bound of course depends on how the skeletons were chosen at each step.

If the bound on the skeleton size is $b$, the maximum degree possible for the tree is $\Delta$, and the support of the potential-function distribution for every level is at most $k$, then the complexity of the computation of the survey is $O(n (kb)^\Delta)$. 
The exponential complexity in $\Delta$ is in practice prohibitive, because in order to obtain surveys that give good bounds for the probability of reconstruction, $b$ may have to be large. However, the important improvement here is that while the exact computation is exponential in $n$, computing the surveys takes time linear in $n$.

A few more remarks regarding implementation are in order:
\begin{enumerate}
\item 
\label{it:deg}
It is not impossible to handle the case when the degree distribution has infinite support. The cases of the small degrees can be computed as above, and for large degrees, if the residual distribution is known to be symmetric (for example if the potential functions are symmetric), then one can use the trivial survey, the one whose basis set is the set of basis vectors. The uniform distribution on the basis vectors is a survey of every symmetric distribution.

\item The size of the domains of the variables also influences the complexity of the algorithm. The computation of the function $f$ in the recursive step in general requires time $|D_{v_1}| \times \dots \times |D_{v_r}|$. However the more important factor is that the size of the skeleton may have to grow significantly with the size of the domain in order to obtain the desired bound. This dependence will have to be studied in the context of particular models. 

\item The strategy of choosing the skeleton at every step crucially
  influences the quality of the bounds. It may be that one type
  of skeleton is beneficial in the beginning iterations and a different type in the
  later ones. In our application we used small base sets in the
  beginning, which makes the first iterations faster, and refined the
  base sets (increased their size) progressively. Perhaps strategies
  for choosing the skeleton can be designed based on the current
  distribution (such as sampling a few probability vectors from it),
  but we have not found such a general-purpose strategy that performs
  well. 

\item In order to obtain rigorous results, the computation of the
  survey of the residual distribution has to be carried out with
  rational numbers. However, naturally the size of the denominators
  increases exponentially, thus a ``rounding" step is required at
  every iteration. In the case of symmetric distributions, such as
  those generated in the Potts model, this can be handled in a similar
  way to item \ref{it:deg}. We choose a large $N$ which will be a
  bound on the size of the denominators. At every step the weights of
  all the vectors in the survey are rounded down to the nearest
  allowed rational number, one with denominator $N$, and the remaining
  weight is distributed among the basis vectors. The resulting
  distribution is a survey of the original one. By the symmetry, in
  the resulting distribution all the basis vectors have the same
  weight, therefore their denominator  is at most $|D_v|N$. 
\end{enumerate}

\section{Application to the Potts model on small-degree trees}

Let $T$ be a random infinite tree rooted at the vertex $\rho$ such that 
the number of children $d$ of every vertex is distributed according to
a random variable  
$\mb{d}$ with expected value $\overline{d}$ and maximum possible value
$d_{max}$.  
Let the domain of values that can be assigned to each vertex be denoted by 
$D = \{1, \cdots,q \}$, and we will also call these values {\it colors}. The
channel $M$ on each edge in the Potts model is given by 
\begin{eqnarray*}
M_{i,j} = 
\begin{cases}
1-p & \text{if } i = j,\\
\frac{p}{q-1} & \text{otherwise},\\
\end{cases}
\end{eqnarray*}
where $0 \leq p \leq 1$. 
This channel corresponds to the $q$-state Potts model on the
tree. Denote the resulting configurations of the tree by
$\overline{\sigma}$ and the alphabet at a vertex $v$ by $\sigma_v$. The
Potts model weighs the resulting configurations according to the
Hamiltonian function 
$H(\overline{\sigma}) = 
\sum_{(u,v) \in E(T)} 
\Ind[\sigma_u = \sigma_v]$
which counts the number of edges in which the color on both end points is
the same. On a finite tree, the probability distribution is given by
\begin{eqnarray*}
\Prob(\overline{\sigma}) =  \frac{1}{Z} \exp\left( \beta \sum_{(u,v) \in E(T)}
\Ind[\sigma_u = \sigma_v]\right)  
\end{eqnarray*}
where $Z$ is a normalizing constant and $\beta$ is an inverse
temperature parameter of the Potts model.

The second largest eigenvalue of the matrix
$M$ is denoted by 
\begin{eqnarray*}
\lambda = 1-\frac{pq}{q-1} = \frac{e^{\beta}-1}{e^\beta +q-1}.
\end{eqnarray*}
In line with the terminology for the Potts model,
$\lambda<0$ corresponds to the {\em ferromagnetic regime} while
$\lambda>0$ corresponds to the anti-ferromagnetic regime. The special
case of proper colorings corresponds to $\lambda = 
\frac{-1}{q-1}$. 

The {\em branching number} of an infinite tree is the supremum of the
real numbers $\gamma \geq 1$ such that $T$ admits a positive flow from
the root to infinity, where on every edge $e$, the flow is bounded by
$\gamma^{-\ell_e}$, where $\ell_e$ denotes the number of edges, including
$e$ on the path from $e$ to the root. Note that for the $d$-ary tree,
the branching number is $d$. In the case where each vertex of the tree
has $k$ children with probability $p_k$, it is known that if the
expected number of children $m=\sum_k kp_k >1$, then the branching
number is $m$ almost surely (see \cite{Lyo90}).

The Kesten-Stigum bound \cite{KS66} for the reconstruction problem says
that for a tree with branching number $d$ such that $\lambda d^2>1$
reconstruction holds. For the Potts model,
Mossel and Peres \cite{MP03} have shown that non-reconstruction holds if 
\begin{eqnarray*}
d \frac{q\lambda^2}{2+\lambda(q-2)} < 1
\end{eqnarray*}
and this bound was improved in \cite{MSW07}. 

At $q=3$, for large enough degree, the
Kesten-Stigum bound was recently shown to be sharp in both the
ferromagnetic and antiferromagnetic cases
\cite{Sly09}. However, there is still a gap for small degrees. We
consider the anti-ferromagnetic and the ferromagnetic Potts models
$q=3$, on $2$- and $3$-ary trees, and on the 
tree in which every vertex is chosen with equal probability to be 2 or
3, and show bounds on the threshold for non-reconstruction.

Let $\sigma$ denote a random configuration of a tree $T$ given by the
transition matrix. Recall that for $a \in D$, $L^a(n)$ denotes the random 
coloring of $L(n)$ conditioned on $\sigma_\rho  =a$. In agreement with the 
notation of \cite{Sly09} we define the random variable
\begin{eqnarray*}
X^+(n) = \Prob_{L \sim L^a(n)}(\sigma_\rho = a \ |  L).
\end{eqnarray*}

This is the conditional probabilities of the color $a$ at the root
when the coloring of the vertices at level $n$ 
is chosen conditioned on $\sigma_\rho = a$. Note
that by the symmetry of the channel, the distribution of $X^+(n)$ does not depend on the particular $a \in D$.

Let $Y^+(n) : = X^+(n)  - \frac{1}{q}$ and denote 
\begin{eqnarray*}
x_n = \Ex [Y^+(n)],  \ \ \ and  \ \ \ \
z_n = \Ex [Y^+(n)^2].
\end{eqnarray*}

Here, the expectation is taken over the randomness of the tree and the 
Markov process on the tree (the random coloring). We go back to the unconditional distribution using the following identity of Sly \cite{Sly09}:

\begin{claim}[\cite{Sly09}]The following relations hold:
$$x_n = \Ex \left[\displaystyle\sum_{a=1}^q \left(\Prob(\sigma_\rho = a |
  L(n)) - \frac{1}{q} \right)^2 \right]\geq z_n \geq 0$$ 
\end{claim}

Therefore, clearly, the condition 
\begin{eqnarray*}
\lim_{n \rightarrow \infty}x_n = 0 
\end{eqnarray*}
is equivalent to non-reconstruction and further, if for each $a$,
\begin{eqnarray*} 
\Ex \left[\left(\Prob(\sigma_\rho = a | L(n)) - \frac{1}{q} \right)^2\right]  \leq
\varepsilon,
\end{eqnarray*}
then $x_n \leq q\varepsilon$. 

By expanding the recursion for the expectation $\Ex [Y^+(n)]$ using a Taylor expansion, exactly as in \cite{Sly09}, we obtain a bound on 
$x_{n+1}$ in terms of $x_n$ and $z_n$. Extending the analysis
of \cite{Sly09} to random trees is straightforward. First, we
obtain a relation in terms of a fixed degree $d$ and then take the
expectation over the distribution for the degree $\mb{d}$.

\begin{theorem}
\label{thm:xn}
Let $q,\mb{d},\overline{d}, d_{max}, \lambda, x_n$ and $z_n$ be defined as above. Then,
\begin{eqnarray}
\label{eq:recursive-inequality}
x_{n+1} \leq 
\overline{d} \lambda^2 x_n  & +  &
\displaystyle\sum_{j=2}^{d_{max}} E \left[{\mb{d} \choose j} \right]  \lambda^{2j} 
\left[
\frac{2(q-1)}{q^2}\left(\frac{q(q-3+\lambda)}{q-1} x_n   -
 \frac{q^2\lambda}{q-1} z_n \right)^j  \right. \\ 
& + & \left. 
\frac{(q-1)(q-2)}{q^2} \left(   - 
    \frac{q(3q-6+2\lambda)}{(q-1)(q-2)}x_n +
  \frac{2q^2\lambda}{(q-1)(q-2)} z_n   \right)^j \right. \\
& - & \left.
\frac{2(q-1)}{q} \left( - \frac{q}{q-1}x_n  \right)^j 
\right].
\end{eqnarray}
\end{theorem}

\begin{proof} The degree of the root of the random tree of depth
$n+1$ is $\mb{d}$. We first bound the expectation conditional on $\mb{d}=d$. 
Following the calculations of \cite{Sly09}, we obtain
\begin{eqnarray*}
\Ex [Y^+(n+1)|\mb{d}=d] &\leq& 
\frac{2(q-1)}{q^2}  
\left(1 + \frac{\lambda^2 q (q-3 +\lambda)}{q-1} x_n - \frac{q^2\lambda^3}{q-1} z_n 
\right)^d \\
&& +
\frac{(q-1)(q-2)}{q^2} 
\left(1-\frac{\lambda^2 q (3q-6+2\lambda)}{(q-1)(q-2)} x_n + \frac{2 q^2\lambda^3}
{(q-1)(q-2)} z_n\right)^d \\
&& -
\frac{2(q-1)}{q} 
\left(1-\frac{\lambda^2 q}{q-1} x_n \right) ^d +1-1/q\\
& = & 
d \lambda^2 x_n  
+ \displaystyle\sum_{j=2}^{d} {d \choose j}   \lambda^{2j} 
\left[
\frac{2(q-1)}{q^2}\left(\frac{q(q-3+\lambda)}{q-1} x_n  -
    \frac{q^2\lambda}{q-1} z_n \right)^j  \right. \\ 
\nonumber & & ~~~~~~~~~~~~~~~~~~~~+\left. 
\frac{(q-1)(q-2)}{q^2} \left(   - 
    \frac{q(3q-6+2\lambda)}{(q-1)(q-2)}x_n +
  \frac{2q^2\lambda}{(q-1)(q-2)} z_n   \right)^j \right.\\
\nonumber & & ~~~~~~~~~~~~~~~~~~~~-  \left.
\frac{2(q-1)}{q} \left( - \frac{q}{q-1}x_n \right)^j 
\right].
\end{eqnarray*}

Taking the expectation over the degree we obtain the statement of the theorem.

\end{proof}

We obtain 
the following Corollary by applying Theorem \ref{thm:xn} to particular values of $q$ and degree
distributions $\mb{d}$.  The inequalities are obtained by
optimizing each term in the summation \eqref{eq:recursive-inequality}
separately subject to the constraint that $0 \leq z_n \leq x_n$.

\begin{corollary}
\label{cor:bounds} Let $q=3$. In the ferromagnetic regime,
\begin{enumerate}[1.]
\item If $\mb{d}=2$ with probability 1, and $\lambda>0$ 
\begin{eqnarray*}
x_{n+1} \leq 2\lambda^2 x_n +
  \frac{3}{2} \lambda^4 x_n^2 (2\lambda^2 +4 \lambda + 1)
\end{eqnarray*}
\item If $\mb{d}=3$ with probability 1, and $\lambda>0$ 
\begin{eqnarray*}
x_{n+1} \leq 3\lambda^2 x_n + \frac{9}{2}
  \lambda^4 x_n^2 (2\lambda^2 +4 \lambda + 1) - \frac{9}{2}\lambda^6
  x_n^3 (1+\lambda)^3
\end{eqnarray*}
\item If $\mb{d}=2$ with probability 1/2, $\mb{d}=3$ with probability
  1/2, and $\lambda >0$ 
\begin{eqnarray*}
x_{n+1} \leq \frac{5}{2} \lambda^2 x_n + 3 \lambda^4 x_n^2 (2\lambda^2
+ 4\lambda + 1) - \frac{9}{4} \lambda^6 x_n^3 (1+\lambda)^3
\end{eqnarray*}
\end{enumerate}
In the antiferromagnetic regime,
\begin{enumerate}[4.]
\item If $\mb{d}=3$ with probability 1, and $\lambda = -\frac{1}{2}$ 
\begin{eqnarray*}
x_{n+1} \leq \frac{3}{4} x_n + \frac{63}{32} x_n^2 - \frac{351}{256}x_n^3
\end{eqnarray*}
\end{enumerate}
\end{corollary}

Using the above bounds, we show that non-reconstruction holds in the following cases. 

\begin{theorem}\label{thm:non-rec-threshold}
Let $q=3$. In the ferromagnetic regime,
\begin{enumerate}[1.]
\item if $\mb{d}=2$ with probability 1, non-reconstruction holds for
  $\lambda \leq 0.69$;
\item if $\mb{d}=3$ with probability 1, reconstruction holds for
  $\lambda \leq 0.555$;
\item if $\mb{d}=2$ with probability 1/2, $\mb{d}=3$ with
  probability 1/2, non-reconstruction holds for $\lambda \leq 0.61$.
\end{enumerate}
In the antiferromagnetic regime,
\begin{enumerate}[4.]
\item If $\mb{d}=3$ with probability 1, and $\lambda = -\frac{1}{2}$ (the
  case of proper colorings), there is non-reconstruction.
\end{enumerate}
\end{theorem}

\begin{proof} 
Let $x^* = x^*(\lambda,q,d)$ denote the upper
bound on $x_n$ given by the algorithm when it is run with inputs
$\lambda, q, d$. The values obtained for $x^*$ by the algorithm were
as follows:

\begin{itemize}
\item If $\mb{d}=2$ w.p. 1, $\lambda = 0.69, \ x^* = 0.02939..$.
\item If $\mb{d}=3$ w.p. 1, $\lambda = 0.555, \ x^* = 0.04457..$.
\item If $\mb{d}=2$ or $3$ each w.p. $1/2$, $\lambda = 0.61, \ x^* =
  0.04057..$.
\item If $\mb{d}=3$ w. p. 1 and $\lambda = -1/2$, $x^* = 0.00038..$
\end{itemize}

The values $x^*$ are an upper bound on $x_n$. It can be verified by
substituting the $\lambda$ above into the corresponding inequalities
in  Corollary \ref{cor:bounds} that in each case $x_{n+1} < C x_n$
where $C$ is a constant smaller than 1. This implies
non-reconstruction since in the limit $x_n$ goes to 0.
\end{proof}

For each of these cases the algorithm was run using Maple 12 and with
only integer computations, using the rounding procedure described in
the previous section. The base sets of the skeletons were selected
manually, refining them whenever the resulting bounds stop
improving. The decomposition functions of the skeletons were selected
to minimize the expected total variation distance between the true
vector and its decomposition using the LP solver of Maple 12. Not more
than 100 iterations were needed for every case to obtain the required
bound. The implementation was run on a MacBook with a 2GHz Intel Core
2 Duo processor and 1GB of RAM. The limiting factor is that the last
tens of iterations take hours  because the skeleton size we choose
towards the end is close to 100 (in the case of d=2, 200). It is
reasonable to expect that with more computational power or time each
of these bounds can be improved, although going beyond the bounds of
Formentin and K\"ulske, if they are not tight, may require
significantly better resources.

\begin{table}[t]\label{table:comp}
\centering
\begin{tabular}{|c|c|c|c|c|}
\hline
$d$ & KS \cite{KS66} & MP \cite{MP03} & FK \cite{FK} & Theorem
\ref{thm:non-rec-threshold}  \\
\hline
 2 & 0.7071.. & 0.6666.. & 0.7018..& 0.69\\
\hline
  3 & 0.5773.. & 0.5302.. & 0.5731..& 0.555\\
\hline
 2.5 & 0.6324.. & 0.5873.. & 0.6278.. & 0.61  \\
\hline
\end{tabular}
\vspace{0.1in}
\caption{The Kesten-Stigum upper bound on the non-reconstruction
threshold and the values of $\lambda$ up to which
non-reconstruction has been shown in \cite{MP03}, \cite{FK} and here, 
for $q=3$ in the ferromagnetic regime ($\lambda>0$).}
\end{table}

The values of $\lambda$ we obtain for non-reconstruction  in the
ferromagnetic regime are  
shown in Table \ref{table:comp} for a comparison with previous bounds from
\cite{KS66,MP03,FK} (the bounds of \cite{MSW07} are not explicitly
derived, so we have not included them). The second column is the
Kesten-Stigum bound below which reconstruction is known to hold. 
In all cases we improve the bound of \cite{MP03}.
The anti-ferromagnetic case (3-coloring on the 3-ary tree) is also not
implied by the bounds of \cite{MP03}.

\section*{Acknowledgements} The authors would like to thank Allan Sly and Florent Krzakala for discussions on the topic. E.M. also thanks the Kavli Institute of Theoretical Physics China for a productive stay during the initial stages of this work. 

\bibliographystyle{siam}\bibliography{reconutf}

\begin{thebibliography}{10}

\bibitem{BVVW08}
{\sc N.~Bhatnagar, J.~Vera, E.~Vigoda, and D.~Weitz}, {\em Reconstruction for
  colourings on trees}, preprint, Available at {\tt
  http://arxiv.org/abs/0711.3664},  (2008).

\bibitem{BCMR06}
{\sc C.~Borgs, J.~Chayes, E.~Mossel, and S.~Roch}, {\em The {K}esten-{S}tigum
  reconstruction bound is tight for roughly symmetric binary channels}, in
  Proceedings of the 47th Annual IEEE Symposium on Foundations of Computer
  Science (FOCS), 2006, pp.~518--530.

\bibitem{EKPS00}
{\sc W.~Evans, C.~Kenyon, Y.~Peres, and L.~J. Schulman}, {\em Broadcasting on
  trees and the {I}sing model}, Ann. Appl. Probab., 10 (2000), pp.~410--433.

\bibitem{FK}
{\sc M.~Formentin and C.~K{\"u}lske}, {\em On the {P}urity of the free boundary
  condition {P}otts measure on random trees}.
\newblock ArXiv:0810.0677, 2008.

\bibitem{Geobook}
{\sc H.~Georgii}, {\em Gibbs Measures and Phase Transitions, Vol 9}, Studies in
  Mathematics, de Gruyter, 1988.

\bibitem{KS66}
{\sc H.~Kesten and B.~P. Stigum}, {\em Additional limit theorems for
  indecomposable multidimensional {G}alton-{W}atson processes}, Ann. Math.
  Statist., 37 (1966), pp.~1463--1481.

\bibitem{Lyo90}
{\sc R.~Lyons}, {\em Random walks and percolation on trees}, Ann. Probab., 18
  (1990), pp.~931--958.

\bibitem{MSW07}
{\sc F.~Martinelli, A.~Sinclair, and D.~Weitz}, {\em Fast mixing for
  independent sets, colorings, and other models on trees}, Random Structures
  and Algorithms, 31 (2007), pp.~134--172.
\newblock Submitted to {C}ommunication in {M}athematical {P}hysics. Extended
  abstract appeared in proceedings of 44'th FOCS.

\bibitem{MM06}
{\sc M.~M\'{e}zard and A.~Montanari}, {\em Reconstruction on trees and the spin
  glass transition}, J. Stat. Phys., 124 (2006), pp.~1317--1350.

\bibitem{MPZ02}
{\sc M.~M\'{e}zard, G.~Parisi, and R.~Zecchina}, {\em Analytic and algorithmic
  solution of random satisfiability problems}, Science, 297 (2002),
  pp.~812--815.
\newblock (Scienceexpress published on-line 27-June-2002;
  10.1126/science.1073287).

\bibitem{Mos01}
{\sc E.~Mossel}, {\em Reconstruction on trees: beating the second eigenvalue},
  Ann. Appl. Probab., 11 (2001), pp.~285--300.

\bibitem{MP03}
{\sc E.~Mossel and Y.~Peres}, {\em Information flow on trees}, Ann. Appl.
  Probab., 13 (2003), pp.~817--844.

\bibitem{RU01}
{\sc T.~Richardson and R.~Urbanke}, {\em The capacity of low-density parity
  check codes under message-passing decoding}, IEEE Trans. Info. Theory, 47
  (2001), pp.~599--618.

\bibitem{RUbook}
\leavevmode\vrule height 2pt depth -1.6pt width 23pt, {\em Modern coding
  theory}, Cambridge University Press, 2008.

\bibitem{Sly08}
{\sc A.~Sly}, {\em Reconstruction of random colourings}, Communications of
  Mathematical Physics,  (2008).

\bibitem{Sly09}
\leavevmode\vrule height 2pt depth -1.6pt width 23pt, {\em Reconstruction of
  symmetric {P}otts {M}odels}, in Proceedings of the 41st Annual ACM Symposium
  on Theory of Computing (STOC), 2009.

\bibitem{ZK07}
{\sc L.~Zdeborova and F.~Krzakala}, {\em Phase transitions in the coloring of
  random graphs}, Phys. Rev. E., 76 (2007), p.~031131.

\end{thebibliography}

\end{document}